%% file: Conditional_mutual_information_and_self-commutator.tex
\documentclass[12pt]{article}

\input{preamble.tex}

\begin{document}

\title{\Large Conditional mutual information and self-commutator}

\author{Lin Zhang\footnote{E-mail: godyalin@163.com}\\
  {\it\small Institute of Mathematics, Hangzhou Dianzi University, Hangzhou 310018, PR~China}}

\date{}
\maketitle \mbox{}\hrule\mbox\\
\begin{abstract}

A simpler approach to the characterization of vanishing conditional mutual information is presented. Some remarks are given as well. More specifically, relating the conditional mutual information to a commutator is a very promising approach towards the approximate version of SSA. That is, it is \emph{conjectured} that small conditional mutual information implies small perturbation of quantum Markov chain.

\end{abstract}
\mbox{}\hrule\mbox\\

\section{Introduction}

To begin with, we fix some notations that will be used in this context.
Let $\cH$ be a finite dimensional complex Hilbert space. A
\emph{quantum state} $\rho$ on $\cH$ is a positive semi-definite
operator of trace one, in particular, for each unit vector
$\ket{\psi} \in \cH$, the operator $\rho = \out{\psi}{\psi}$ is said
to be a \emph{pure state}. The set of all quantum states on $\cH$ is
denoted by $\density{\cH}$. For each quantum state
$\rho\in\density{\cH}$, its von Neumann entropy is defined by
$$
\rS(\rho) \defeq - \Tr{\rho\log\rho}.
$$
The \emph{relative entropy} of two mixed states $\rho$ and $\sigma$ is
defined by
$$
\rS(\rho||\sigma) \defeq \left\{\begin{array}{ll}
                             \Tr{\rho(\log\rho -
\log\sigma)}, & \text{if}\ \supp(\rho) \subseteq
\supp(\sigma), \\
                             +\infty, & \text{otherwise}.
                           \end{array}
\right.
$$
A \emph{quantum channel} $\Phi$ on $\cH$ is a trace-preserving completely positive linear
map defined over the set $\density{\cH}$. It follows that there exists linear
operators $\set{K_\mu}_\mu$ on $\cH$ such that $\sum_\mu
K^\dagger_\mu K_\mu = \I$ and $\Phi = \sum_\mu \mathrm{Ad}_{K_\mu}$,
that is, for each quantum state $\rho$, we have the Kraus
representation
\begin{eqnarray*}
\Phi(\rho) = \sum_\mu K_\mu \rho K^\dagger_\mu.
\end{eqnarray*}

The celebrated strong subadditivity (SSA) inequality of quantum
entropy, proved by Lie and Ruskai in
\cite{Lieb1973},
\begin{eqnarray}\label{eq:SSA-1} \rS(\rho_{ABC}) +
\rS(\rho_B) \leqslant \rS(\rho_{AB}) + \rS(\rho_{BC}),
\end{eqnarray}
is a very powerful tool in quantum information theory. Recently, the operator extension of SSA is obtained by Kim in \cite{Kim2012}. Following the line of Kim, Ruskai gives a family of new operator inequalities in \cite{Ruskai2012}.

Conditional mutual information, measuring the correlations of two quantum
systems relative to a third, is defined as follows: Given a tripartite state $\rho_{ABC}$, it is defined as
\begin{eqnarray}
I(A:C|B)_\rho \defeq \rS(\rho_{AB})+ \rS(\rho_{BC}) - \rS(\rho_{ABC}) - \rS(\rho_B).
\end{eqnarray}
Clearly conditional mutual information is nonnegative by SSA.

Ruskai is the first one to discuss the equality condition of SSA. By analyzing the equality condition of Golden-Thompson inequality, she obtained the following characterization \cite{Ruskai2002}:
\begin{eqnarray}
I(A:C|B)_\rho = 0 \Longleftrightarrow \log\rho_{ABC} + \log\rho_B = \log\rho_{AB} + \log\rho_{BC}.
\end{eqnarray}

Later on, using the relative modular approach established by Araki, Petz gave another characterization of the equality condition of SSA \cite{Petz2003}:
\begin{eqnarray}
I(A:C|B)_\rho = 0 \Longleftrightarrow \rho^{\mathrm{i}t}_{ABC}\rho^{-\mathrm{i}t}_{BC} = \rho^{\mathrm{i}t}_{AB} \rho^{-\mathrm{i}t}_B\quad(\forall t\in\real),
\end{eqnarray}
where $\mathrm{i} = \sqrt{-1}$ is the imaginary unit.

Hayden \emph{et al.} in \cite{Hayden2004} showed that $I(A:C|B)_\rho =0$ if and only if the following
conditions hold:
\begin{enumerate}[(i)]
\item $\cH_B = \bigoplus_k \cH_{b^L_k} \ot \cH_{b^R_k}$,
\item $\rho_{ABC} = \bigoplus_k  p_k \rho_{Ab^L_k} \ot \rho_{b^R_kC}$, where $\rho_{Ab^L_k}\in\density{\cH_A \ot \cH_{b^L_k}}, \rho_{b^R_kC} \in \density{\cH_{b^R_k} \ot\cH_C}$
for each index $k$; and $\set{p_k}$ is a probability distribution.
\end{enumerate}

In \cite{Brandao2011}, Brand\~{a}o \emph{et al.} first obtained the following lower bound for $I(A:C|B)_\rho$:
\begin{eqnarray}
I(A:C|B)_\rho \geqslant \frac1{8\ln 2} \min_{\sigma_{AC}\in\mathbb{SEP}}\norm{\rho_{AC} - \sigma_{AC}}^2_{1-\mathbb{LOCC}},
\end{eqnarray}
where
$$
\norm{\rho_{AC} - \sigma_{AC}}^2_{1-\mathbb{LOCC}} \defeq \sup_{\cM\in1-\mathbb{LOCC}}\norm{\cM(\rho_{AC}) - \cM(\sigma_{AC})}_1.
$$
Based on this result, he cracked a \emph{long-standing} open problem in quantum information theory. That is, the squashed entanglement is \emph{faithful}. Later, Li in \cite{Li2012} gave another approach to study the same problem and improved the lower bound for $I(A:C|B)_\rho$:
\begin{eqnarray}
I(A:C|B)_\rho \geqslant \frac1{2\ln 2} \min_{\sigma_{AC}\in\mathbb{SEP}}\norm{\rho_{AC} - \sigma_{AC}}^2_{1-\mathbb{LOCC}}.
\end{eqnarray}

Along with the above line, Ibinson \emph{et al.} in \cite{Ibinson2008} studied the robustness of quantum Markov chains, i.e. the perturbation of states of vanishing conditional mutual information. In order to study it further, We need to employ the following famous characterization of saturation of monotonicity inequality of relative entropy.

\begin{thrm}[Petz, \cite{Petz1988,Hiai2011}]\label{th:relative-entropy-petz}
Let $\rho,\sigma\in\density{\cH}$, $\Phi$ be a quantum channel defined over $\cH$. If $\supp(\rho)\subseteq\supp(\sigma)$, then
\begin{eqnarray}
\rS(\rho||\sigma) = \rS(\Phi(\rho)||\Phi(\sigma))\quad\text{if and only if}\quad \Phi^\dagger_\sigma\circ\Phi(\rho) = \rho,
\end{eqnarray}
where $\Phi^\dagger_\sigma = \mathrm{Ad}_{\sigma^{1/2}}\circ\Phi^\dagger\circ\mathrm{Ad}_{\Phi(\sigma)^{-1/2}}$.
\end{thrm}
Noting that the equivalence between monotonicity of relative entropy and SSA, the above theorem, in fact, gives another characterization of vanishing conditional mutual information of quantum states.

\section{Main result}

In this section, we give another characterization of saturation of SSA from the perspective commutativity.
\begin{thrm}\label{th:SSA-with-equality}
Let $\rho_{ABC}\in\density{\cH_A\ot\cH_B\ot\cH_C}$. Denote
\begin{eqnarray*}
M &\defeq& (\rho^{1/2}_{AB}\ot\I_C)(\I_A\ot\rho^{-1/2}_B\ot\I_C)(\I_A\ot\rho^{1/2}_{BC})\\
&\equiv& \rho^{1/2}_{AB}\rho^{-1/2}_B\rho^{1/2}_{BC}.
\end{eqnarray*}
Then the following conditions are equivalent:
\begin{enumerate}[(i)]
\item The conditional mutual information is vanished, i.e. $I(A:C|B)_\rho = 0$;
\item $\rho_{ABC} = MM^\dagger = \rho^{1/2}_{AB}\rho^{-1/2}_B\rho_{BC}\rho^{-1/2}_B \rho^{1/2}_{AB}$;
\item $\rho_{ABC} = M^\dagger M = \rho^{1/2}_{BC}\rho^{-1/2}_B\rho_{AB}\rho^{-1/2}_B \rho^{1/2}_{BC}$;
\end{enumerate}
\end{thrm}

\begin{proof}
Clearly, the conditional mutual information is vanished, i.e. $I(A:C|B)_\rho = 0$, if and only if
\begin{eqnarray}
\rS(\rho_{ABC}) + \rS(\rho_B) = \rS(\rho_{AB}) + \rS(\rho_{BC}).
\end{eqnarray}
Hence we have that
\begin{eqnarray}
\rS(\rho_{AB}||\rho_A\ot\rho_B) &=& \rS(\rho_{ABC}||\rho_A\ot\rho_{BC}),\\
\rS(\rho_{BC}||\rho_B\ot\rho_C) &=& \rS(\rho_{ABC}||\rho_{AB}\ot\rho_C).
\end{eqnarray}
Now let $\Phi=\trace_C$ and $\Psi=\trace_A$, it follows that
\begin{eqnarray}
\rS(\rho_{ABC}||\rho_A\ot\rho_{BC}) &=& \rS(\Phi(\rho_{ABC})||\Phi(\rho_A\ot\rho_{BC})),\label{eq:one-1}\\
\rS(\rho_{ABC}||\rho_{AB}\ot\rho_C) &=& \rS(\Psi(\rho_{ABC})||\Psi(\rho_{AB}\ot\rho_C)).\label{eq:one-2}
\end{eqnarray}
By Theorem~\ref{th:relative-entropy-petz}, we see that both Eq.~\eqref{eq:one-1} and Eq.~\eqref{eq:one-2} hold if and only if
\begin{eqnarray}
\rho_{ABC} = \Phi^\dagger_{\rho_A\ot\rho_{BC}} \circ\Phi(\rho_{ABC})\quad\text{and}\quad \rho_{ABC} = \Psi^\dagger_{\rho_{AB}\ot\rho_C} \circ\Psi(\rho_{ABC}),
\end{eqnarray}
i.e.
\begin{eqnarray}
\rho_{ABC} = \rho^{1/2}_{AB}\rho^{-1/2}_B\rho_{BC}\rho^{-1/2}_B \rho^{1/2}_{AB} = \rho^{1/2}_{BC}\rho^{-1/2}_B\rho_{AB}\rho^{-1/2}_B \rho^{1/2}_{BC}.
\end{eqnarray}
This amounts to say that $I(A:C|B)_\rho = 0$ if and only if $\rho_{ABC} = MM^\dagger = M^\dagger M$.
\end{proof}

\begin{remark}
In \cite{Leifer2008}, Leifer and Poulin gave a condition which is equivalent to our result. There they mainly focus on the characterization of conditional independence in terms of noncommutative probabilistic language by analogy with classical conditional independence. By combining the Lie-Trotter product formula:
\begin{eqnarray}
\exp(A+B) &=& \lim_{n\to\infty}\Br{\exp(A/n)\exp(B/n)}^n \\
&=& \lim_{n\to\infty}\Br{\exp(A/{2n})\exp(B/n)\exp(A/{2n})}^n,
\end{eqnarray}
where both $A$ and $B$ are square matrices of the same order, a characterization of vanishing conditional mutual information was obtained. Clearly, the Lie-Trotter product formula is not easy to deal with. In fact, our proof is more naturally and much simpler than that of theirs.
\end{remark}

In the following, we denote by $\Br{X,X^\dagger}$ the \emph{self-commutator} of an operator or a matrix $X$.

\begin{cor}\label{cor:zero-SSA-commutator}
With the notation mentioned above in Theorem~\ref{th:SSA-with-equality}, the following statement is true:
$I(A:C|B)_\rho = 0$ implies $\Br{M,M^\dagger} = 0$. In other words, $\Br{M,M^\dagger} \neq 0$ implies $I(A:C|B)_\rho \neq 0$.
\end{cor}

\begin{proof}
We assume that $I(A:C|B)_\rho = 0$. From Theorem~\ref{th:SSA-with-equality}, we know that $\rho_{ABC} = MM^\dagger = M^\dagger M$, implying $\Br{M,M^\dagger} = 0$.
\end{proof}

\section{Discussion}

A natural question arises: Can we derive $I(A:C|B)_\rho = 0$ from
$\Br{M,M^\dagger} = 0$? The answer is no! Indeed, we know from the
discussion in \cite{Winter} that, if the operators
$\rho_{AB},\rho_{BC}$ and $\rho_B$ all commute, then
\begin{eqnarray}\label{eq:Pinsker}
I(A:C|B)_\rho = \rS(\rho_{ABC}||MM^\dagger).
\end{eqnarray}
Now let $\rho_{ABC} = \sum_{i,j,k} p_{ijk}\out{ijk}{ijk}$ with
$\set{p_{ijk}}$ being an arbitrary joint probability distribution.
Thus
$$
MM^\dagger = M^\dagger M = \sum_{i,j,k}
\tfrac{p_{ij}p_{jk}}{p_j}\out{ijk}{ijk},
$$
where $p_{ij}=\sum_k p_{ijk},p_{jk}=\sum_ip_{ijk}$ and
$p_j=\sum_{i,k}p_{ijk}$ are corresponding marginal distributions,
respectively. But in general,
$p_{ijk}\neq\tfrac{p_{ij}p_{jk}}{p_j}$. Therefore we have a specific
example in which $[M,M^\dagger]=0$, and $\rho_{ABC}\neq MM^\dagger$,
i.e. $I(A:C|B)_\rho>0$. By employing the Pinsker's inequality to
Eq.~\eqref{eq:Pinsker}, it follows in this special case that
$$
I(A:C|B)_\rho \geqslant \frac1{2\ln2}\norm{\rho_{ABC} -
MM^\dagger}^2_1.
$$

Along with the above line, all tripartite states can be classified
into three categories:
$$
\density{\cH_A\ot\cH_B\ot\cH_C} = D_1\cup D_2 \cup D_3,
$$
where
\begin{enumerate}[(i)]
\item $D_1 \defeq \Set{\rho_{ABC}:\rho_{ABC} = MM^\dagger, [M,M^\dagger]=0}$.
\item $D_2 \defeq \Set{\rho_{ABC}:\rho_{ABC} \neq MM^\dagger, [M,M^\dagger]=0}$.
\item $D_3 \defeq \Set{\rho_{ABC}:[M,M^\dagger]\neq0}$.
\end{enumerate}
For related topics please refer to \cite{Poulin2011,Brown2012}. The
result obtained in Corollary~\ref{cor:zero-SSA-commutator} can be
employed to discuss a small conditional mutual information. I. Kim
\cite{Kim} tries to give a universal proof of the following
inequality:
$$
I(A:C|B)_\rho \geqslant \frac1{2\ln2}\norm{\rho_{ABC} - MM^\dagger}^2_1.
$$
As a matter of fact, if the above inequality holds, then a similar inequality holds:
$$
I(A:C|B)_\rho \geqslant \frac1{2\ln2}\norm{\rho_{ABC} - M^\dagger M}^2_1.
$$
The validity or non-validity of both inequalities can be guaranteed
by Theorem~\ref{th:SSA-with-equality}. According to the numerical
computation by Kim, up to now, there are no states violating these
inequalities. Therefore we have the following \emph{conjecture}:
\begin{eqnarray}\label{conjecture}
I(A:C|B)_\rho \geqslant \frac1{2\ln2}\max\Set{\norm{\rho_{ABC} - MM^\dagger}^2_1,\norm{\rho_{ABC} - M^\dagger M}^2_1}.
\end{eqnarray}
We can connect the total amount of conditional mutual information contained in the tripartite state $\rho_{ABC}$ with the trace-norm of the commutator $\Br{M,M^\dagger}$ as follows: if the above conjecture holds, then we have
\begin{eqnarray}
I(A:C|B)_\rho \geqslant \frac1{8\ln2}\norm{\Br{M, M^\dagger}}^2_1,
\end{eqnarray}
but not vice versa. Even though the above conjecture is false, it is
still possible that this inequality is true.

In \cite{Winter}, the authors proposed the following
question: For a given quantum channel $\Phi\in\trans{\cH_A,\cH_B}$ and
states $\rho,\sigma\in\density{\cH_A}$, does there exist a quantum
channel $\Psi\in\trans{\cH_B,\cH_A}$ with
$\Psi\circ\Phi(\sigma)=\sigma$ and
\begin{eqnarray}
\rS(\rho||\sigma) \geqslant \rS(\Phi(\rho)||\Phi(\sigma)) +
\rS(\rho||\Psi\circ\Phi(\rho))?
\end{eqnarray}
The authors affirmatively answer this question in the classical
case. The quantum case is still open. Although the authors proved that the following inequality is not valid in general:
$$
\rS(\rho||\sigma) \ngeqslant \rS(\Phi(\rho)||\Phi(\sigma)) +
\rS(\rho||\Phi^\dagger_\sigma\circ\Phi(\rho))
$$
However, the following inequality may still be correct:
$$
\rS(\rho||\sigma) \geqslant \rS(\Phi(\rho)||\Phi(\sigma)) +
\frac1{2\ln2}\norm{\rho - \Phi^\dagger_\sigma\circ\Phi(\rho)}^2_1.
$$
In fact, if this modified inequality holds, then Eq.~\eqref{conjecture} will hold.

A future research may be directed to establish a connection between
conditional mutual information and almost commuting normal matrices
\cite{Rordam}.

\subsection*{Acknowledgement}

This project is supported by the Research Program of Hangzhou Dianzi University (KYS075612038). LZ acknowledges Matthew Leifer, David Poulin and Mark M. Wilde for drawing my attention to the References \cite{Leifer2008,Poulin2011}. The author also would like to thank Isaac H. Kim for valuable discussions. Minghua Lin's remarks are useful as well. Especially, thank F. Brand\~{a}o for drawing my attention to the problem concerning the approximate version of vanishing conditional mutual information.



\end{document}

%% file: preamble.tex
\usepackage[T1]{fontenc}
\usepackage[sc]{mathpazo}
\usepackage{amsmath}
\usepackage{amssymb}
\usepackage{enumerate}
\usepackage{amsthm}
\usepackage{amsfonts,mathrsfs}
\usepackage[style]{fncychap}
\usepackage{graphicx} 
\usepackage{geometry} 
\usepackage{eepic}
\usepackage{ifthen}
\newboolean{ElectronicVersion}
\setboolean{ElectronicVersion}{true} 

\usepackage{hyperref}
\hypersetup{pdfpagemode=UseNone}

\newcommand{\tinyspace}{\mspace{1mu}}

\newcommand{\norm}[1]{\left\lVert\tinyspace #1 \tinyspace\right\rVert}

\newcommand{\setft}[1]{\mathrm{#1}}

\newcommand{\density}[1]{\setft{D}\left(#1\right)}

\newcommand{\trans}[1]{\setft{T}\left(#1\right)}

\newcommand{\supp}{{\operatorname{supp}}}

\def\real{\mathbb{R}}

\def\I{\mathbb{1}}

\newenvironment{mylist}[1]{\begin{list}{}{
    \setlength{\leftmargin}{#1}
    \setlength{\rightmargin}{0mm}
    \setlength{\labelsep}{2mm}
    \setlength{\labelwidth}{8mm}
    \setlength{\itemsep}{0mm}}}
    {\end{list}}

\def\ot{\otimes}

\newcommand{\out}[2]{| #1\rangle\langle #2 |}

\newcommand{\defeq}{\stackrel{\smash{\textnormal{\tiny def}}}{=}}

\newcommand{\Pa}[1]{\left(#1\right)}

\newcommand{\Br}[1]{\left[#1\right]}
\newcommand{\set}[1]{\{#1\}}
\newcommand{\Set}[1]{\left\{#1\right\}}

\newcommand{\ket}[1]{|#1\rangle}

\DeclareMathOperator{\trace}{Tr}

\newcommand{\Ptr}[2]{\trace_{#1}\Pa{#2}}

\newcommand{\Tr}[1]{\Ptr{}{#1}}

\def\cH{\mathcal{H}}
\def\cM{\mathcal{M}}

\def\rS{\mathrm{S}}

\newtheorem{thrm}{Theorem}[section]

\newtheorem{cor}[thrm]{Corollary}

\theoremstyle{definition}

\newtheorem{remark}[thrm]{Remark}

\numberwithin{equation}{section}

\newcounter{questionnumber}